\documentclass{article}

\usepackage[round,authoryear]{natbib}

\usepackage{amssymb,amsfonts,amsmath,amsthm}

\usepackage{graphicx}  
\usepackage{subfig}    
\usepackage{booktabs}  

\newcommand{\poi}{\mathrm{Poi}}
\newcommand{\rd}{\mathrm{\, d}}

\newcommand{\e}{\mathbb{E}}
\renewcommand{\Pr}{\mathbb{P}}

\newtheorem{theorem}{Theorem}

\newtheorem{lemma}{Lemma}
\newtheorem{corollary}{Corollary}

\begin{document}

\title{Correct ordering in the Zipf--Poisson ensemble}
\date{September 2010}
\author{Justin S. Dyer\\
Stanford University\\Stanford, CA, USA
\and
Art B. Owen\\
Stanford University\\Stanford, CA, USA
}

\maketitle

\begin{abstract}
We consider a Zipf--Poisson ensemble
in which $X_i\sim\poi(Ni^{-\alpha})$
for $\alpha>1$ and $N>0$ and integers $i\ge 1$. 
As $N\to\infty$ the first $n'(N)$
random variables have their proper order $X_1>X_2>\dots>X_{n'}$
relative to each other,
with probability tending to $1$
for $n'$ up to $(AN/\log(N))^{1/(\alpha+2)}$
for an explicit constant $A(\alpha)\ge 3/4$.
The rate $N^{1/(\alpha+2)}$ cannot be achieved.
The ordering of the first $n'(N)$ entities
does not preclude
$X_m>X_{n'}$ for some interloping $m>n'$. The first
$n''$ random variables are correctly ordered exclusive of any interlopers,
with probability tending to $1$
if $n''\le (BN/\log(N))^{1/(\alpha+2)}$ 
for $B<A$.
For a Zipf--Poisson model of the British National Corpus,
which has a total word count of
$100{,}000{,}000$, our result estimates that the
$72$ words with the highest counts
are properly ordered.
\end{abstract}


\section{Introduction}

Power law distributions are ubiquitous,
arising in studies of degree distributions
of large networks, book and music sales counts,
frequencies of words in literature and even baby names.
It is common that the relative frequency
of the $i$'th most popular term
falls off roughly as $i^{-\alpha}$ for
a constant $\alpha$ slightly larger than $1$.
This is the pattern made famous
by \cite{zipf:1949}. 
Data usually show some deviations from a
pure Zipf model. 
The Zipf--Mandelbrot law for which
the $i$'th frequency is proportional to
$(i+k)^{-\alpha}$ where $k\ge 0$ is often a much
better fit. That and many other models
are described in \citet[Chapter 9]{pope:2009}.

The usual methods for fitting long tailed
distributions assume an IID sample.
However, in many applications 
a persistent set of entities is
repeatedly sampled under slightly
different conditions.  For example, if
one gathers a large sample of English
text, the word `the' will be the most
frequent word with overwhelming probability.
No other word has such high probability
and repeated sampling will not give
either zero or two instances of such a
very popular word.
Similarly in Internet applications, the
most popular URLs in one sample are
likely to reappear in a similar sample,
taken shortly thereafter or in a closely
related stratum of users.  The movies
most likely to be rated at Netflix in
one week will, with only a few changes,
be the same popular movies the next week.

Because the entities themselves have a
meaning beyond our sample data,
it is natural to wonder whether they
are in the correct order in our sample.
The problem we address here is the
error in ranks estimated from count data.
By focussing on count data we are
excluding other long tailed data such
as the lengths of rivers.

In a large data set, the top few most popular
items are likely to be correctly identified
while the items that appeared only a handful
of times cannot be confidently ordered
from the sample. 
We are interested in drawing the line
between ranks that are well identified
and those that may be subject to sampling
fluctuations.
One of our motivating applications is a graphical
display in~\cite{dyer:owen:2010:tr:2}.
Using that display one is able to depict
a head to tail affinity for movie
ratings data:  the busiest raters are over represented
in the most obscure movies and conversely
the rare raters are over represented in ratings
of very frequently rated movies. Both effects
are concentrated in a small corner
of the diagram.
The graphic has greater importance
if it applies to a source generating
the data than if it applies only
to the data at hand.

This paper uses the Zipf law because it is
the simplest model for long tailed rank data
and we can use it to get precisely
stated asymptotic
results. If instead we are given another
distribution, then a numerical method
described in Section~\ref{sec:discussion}
is very convenient to apply.

If we  
suppose that the item counts are 
independent and Poisson
distributed with expectations that follow
a power law, then a precise answer is possible.
We define
the Zipf--Poisson ensemble to be an infinite
collection of independent random variables
$X_i\sim\poi(\lambda_i)$
where $\lambda_i = Ni^{-\alpha}$
for parameters $\alpha>1$ and $N>0$.
Our main results are summarized in Theorem~\ref{thm:zpok} below.

\begin{theorem}\label{thm:zpok}
Let $X_i$ be sampled from the Zipf--Poisson ensemble
with parameter $\alpha>1$.
If $n=n(N)\le (AN/\log(N))^{1/(\alpha+2)}$
for $A=\alpha^2(\alpha+2)/4$,
then
\begin{align}\label{eq:zpok1}
\lim_{N\to\infty}\Pr\bigl( X_1>X_2>\cdots>X_n\bigr)=1.
\end{align}
If $n=n(N)\le (BN/\log(N))^{1/(\alpha+2)}$ for $B < A$,
then
\begin{align}\label{eq:zpok2}
\lim_{N\to\infty}\Pr\bigl( X_1>X_2>\cdots>X_n>\max_{i > n}X_i\bigr)=1.
\end{align}
If $n=n(N)\ge CN^{1/(\alpha+2)}$ for any $C>0$, then
\begin{align}\label{eq:zpok3}
\lim_{N\to\infty}\Pr\bigl( X_1>X_2>\cdots>X_n\bigr)=0.
\end{align}
\end{theorem}

Equation~\eqref{eq:zpok1} states that the
top $n'=\lfloor (AN/\log(N))^{1/(\alpha+2)}\rfloor$
entities,
with $A=\alpha^2(\alpha+2)/4$,
are correctly ordered among themselves
with probabilty tending to $1$ as $N\to\infty$.
From $\alpha>1$ we have $A>3/4$.
Equation~\eqref{eq:zpok3} shows that
we cannot remove $\log(N)$ from the denominator,
because the first $CN^{1/(\alpha+2)}$ entities will
fail to have the correct joint ordering with
a probability approaching $1$ as $N\to\infty$.

Equation~\eqref{eq:zpok1} leaves open the possibility
that some entity beyond the $n'$th manages to get among
the top $n'$ entities due to sampling fluctuations. 
Those entities each have
only a small chance to be bigger than $X_{n'}$, 
but there are infinitely 
many of them.  Equation~\eqref{eq:zpok2}
shows that with probability tending to $1$,
the first 
$n''=\lfloor (BN/\log(N))^{1/(\alpha+2)}\rfloor$
entities are the correct first $n''$
entities in the correct order.
The limit holds for any $B<A$. 
That is, there is very little scope for interlopers.

Section~\ref{sec:bnc} shows an example
based on $100{,}000{,}000$ words
of the British National Corpus (BNC).
See \cite{asto:burn:1998}.
Using $\alpha$ near $1.1$ in the asymptotic formulas,
we estimate that the first $72$
words are correctly ordered among themselves.
In a Monte Carlo simulation, very few
interloping counts were seen. The estimate
$n'=72$ depends on the Zipf--Poisson
assumption which is an idealization,
but it is quite stable if the log--log
relationship is locally linear in a critical
region of $n$ values.

Section~\ref{sec:proofs} proves our results.
Of independent interest there is 
Lemma~\ref{lem:skellbound} which
gives a Chernoff bound for the \cite{skel:1946}
distribution: For $\lambda\ge\nu>0$ we show that
$\Pr( \poi(\lambda)\le\poi(\nu))\le\exp(-(\sqrt{\lambda}-\sqrt{\nu})^2)$
where
$\poi(\lambda)$ and $\poi(\nu)$ are independent
Poisson random variables with the given means.
Section~\ref{sec:discussion} has our conclusions.

\section{Example: the British National Corpus}\label{sec:bnc}

Figure~\ref{fig:kilg} 
plots the frequency
of English words versus
their rank on a log-log scale, for all words
appearing at least $800$ times among
the approximately $100$ million words of the BNC.
The counts are from \cite{kilg:2006}.
The data plotted have a nearly
linear trend with a slope just steeper
than $-1$.   They are not perfectly
Zipf-like, but the fit is extraordinarily
good considering that it uses just
one parameter for $100$ million total words.

The top $10$ counts from Figure~\ref{fig:kilg}
are shown in Table~\ref{tab:topten}.
The most frequent word `the' is much
more frequent than the second most frequent
word `be'.  
The process generating this data 
clearly favors the word `the' over `be'
and a $p$-value for whether these words
might be equally frequent, using Poisson
assumptions is overwhelmingly significant.
Though the $9$'th and $10$'th words
have counts that are within a few percent
of each other, they too are significantly
different, as judged by
$(X_{9}-X_{10})/\sqrt{X_9+X_{10}}\doteq 34.9$,
the number of estimated standard deviations
separating them.
The $500$'th and $501$'st most popular
words are `report' and `pass' with
counts of 20{,}660 and 20{,}633
respectively. These are not significantly
different.

\begin{figure}
\includegraphics[width=\hsize]{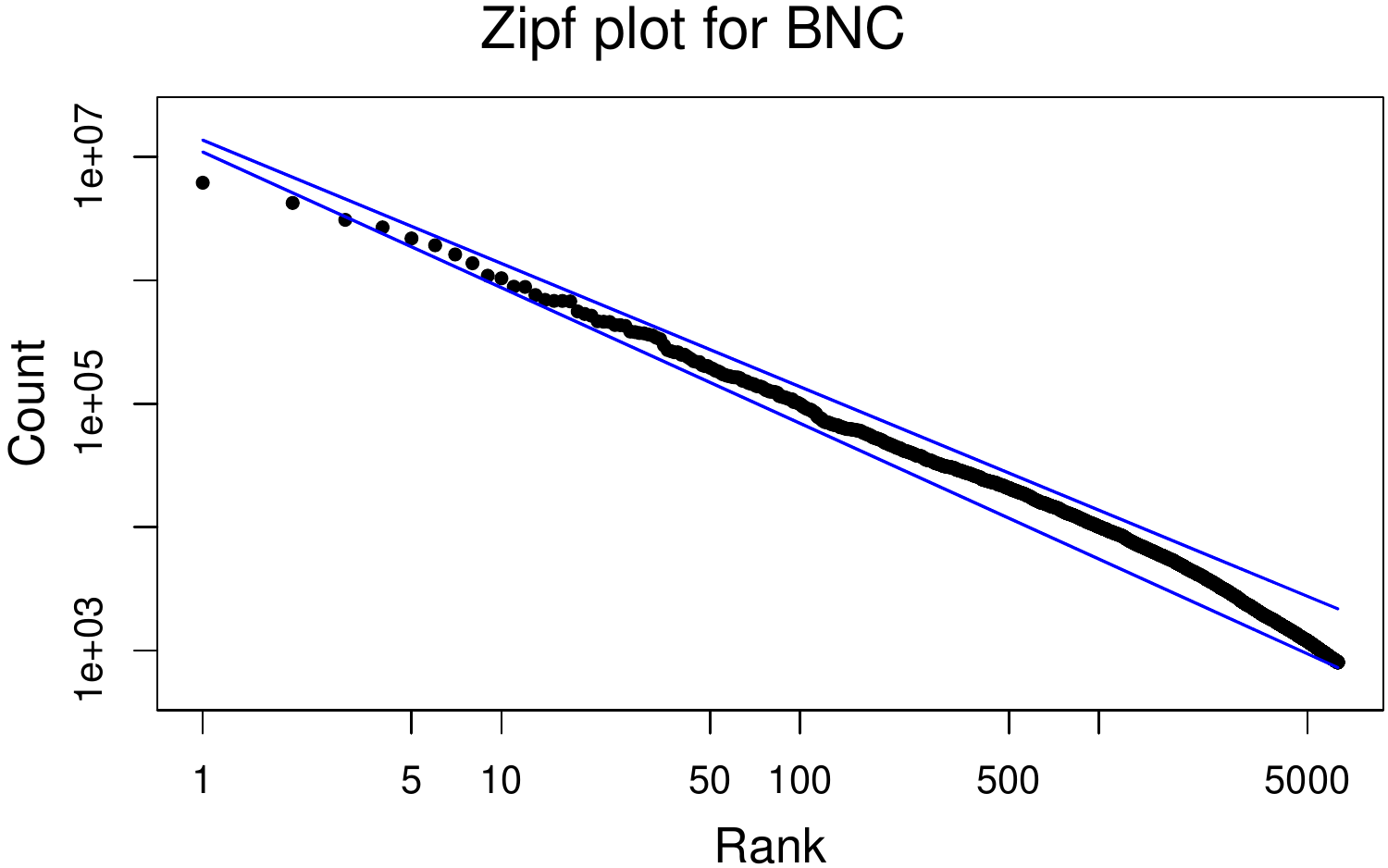}
\caption{\label{fig:kilg}
Zipf plot for the British National
Corpus data. The reference line
above the data has slope $-1$, while that below
the data has slope $-1.1$.
}
\end{figure}

\begin{table}\centering
\newcommand{\pho}{\phantom{1}}
\begin{tabular}{ccc}
\toprule
Rank & Word & Count\\
\midrule 
\pho1 &  the & 6{,}187{,}267 \\
\pho2 &   be & 4{,}239{,}632 \\
\pho3 &   of & 3{,}093{,}444 \\
\pho4 &  and & 2{,}687{,}863  \\
\pho5 &    a & 2{,}186{,}369 \\
\pho6 &   in & 1{,}924{,}315 \\
\pho7 &   to & 1{,}620{,}850 \\
\pho8 & have & 1{,}375{,}636 \\
\pho9 &   it & 1{,}090{,}186 \\
10 &   to & 1{,}039{,}323\\
\toprule
\end{tabular}
\caption{\label{tab:topten}
The top ten most frequent words
from the British National Corpus,
with their frequencies.
Item 7 is the word `to', used as
an infinitive marker, while item 10
is `to' used as a preposition.
In ``I went to the library to read.''
the first `to' is a preposition and
the second is an infinitive marker.
}
\end{table}

We will use a value of $\alpha$ close to $1.1$ to illustrate
the results of Theorem~\ref{thm:zpok}.
The data appear to have approximately
this slope in what will turn out to
be the important region, with ranks from $10$
to $100$.
We don't know $N$ but we can estimate it.
Let
$T= \sum_{i=1}^\infty X_i$
be the total count.
Then $\e(T)=\sum_{i=1}^\infty Ni^{-\alpha}=N\zeta(\alpha)$
where $\zeta(\cdot)$ is the Riemann
zeta function.
We find that
$\zeta(\alpha_*)=10$
for $\alpha_*\doteq1.106$.
Choosing $\alpha=\alpha_*$ we find that 
$T=10^8$ corresponds to $N=N_*\equiv10^7$.

Theorem~\ref{thm:zpok} has the
top $n'=(A(\alpha)N/(\log(N)))^{1/(\alpha+2)}$ entities
correctly ordered among themselves
with probability tending to $1$.
For the BNC data we get
$n'  =(A(\alpha_*)N_*/(\log(N_*)))^{1/(\alpha_*+2)}\doteq 72.08$.
For data like this, we could reasonably expect the
true $72$ most popular words to be correctly ordered
among themselves. 

We did a small simulation of the Zipf--Poisson
model. 
The results are shown
in Figure~\ref{fig:n2hist}.
The number of correctly ordered items
ranged from $69$ to $153$ in those $1000$
simulations. The number was only 
smaller than $72$ for $2$ of the simulated cases.

In our simulation, the first rank
error to occur was usually a transposition between
the $n$'th and $n+1$'st entity. This happened
$982$ times.  There were $7$ cases with
a tie between the $n$'th and $n+1$'st
entity. The remaining $11$ cases all involved
the $n+2$'nd entity getting ahead of the $n$'th.
As a result, we see that interlopers are very
rare, as we might expect from Lemma~\ref{lem:youshallnotpass}.

\begin{figure}[t]
\includegraphics[width=\hsize]{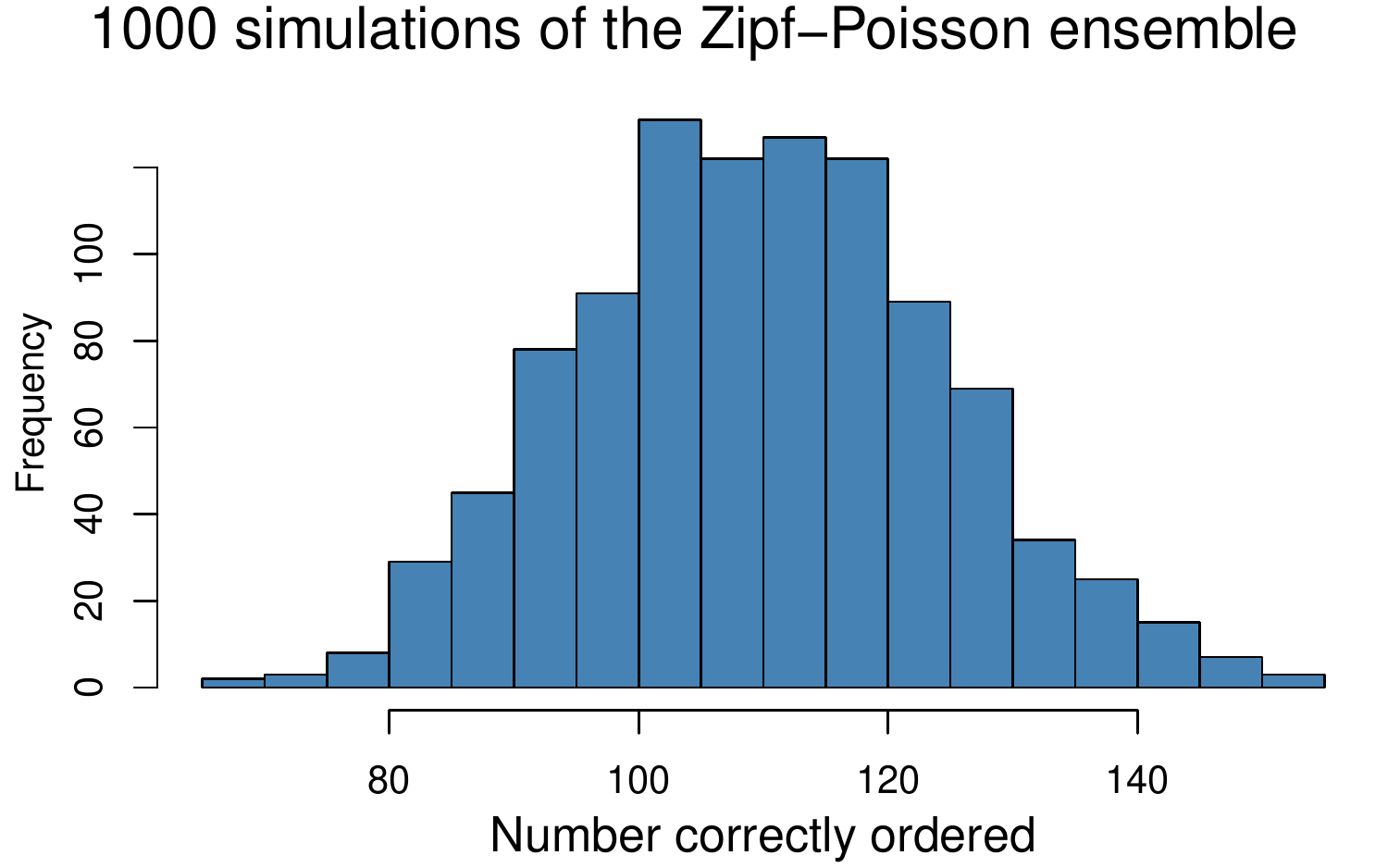}
\caption{\label{fig:n2hist}
The Zipf--Poisson ensemble with $N=10^7$ and
$\alpha=1.106$ was simulated $1000$
times.  
The histogram shows the distribution
of the number of correctly ordered words.
}
\end{figure}

Lack of fit of the Zipf--Poisson model will
affect the estimate somewhat.
Here we give a simple analysis to 
show that the estimate $n'=72$
is remarkably stable.
Even though the log--log plot in Figure~\ref{fig:kilg}
starts off shallower than slope $-1$, the first
$10$ counts are so large that we can be confident
that they are correctly ordered.
Similarly, the log--log plot ends up somewhat steeper
than $-1.1$, but that takes place for very rare
items that have negligible chance of ending up ahead
of the $72$'nd word.
As a result we choose to work with $\alpha=1.106$
and re-estimate $N$ to match the counts in
the range $10\le n\le 100$.
Those counts are large and very stable. A simple estimate of
$N$ is $N_i=X_ii^\alpha$. 
For this data $\min_{10\le i\le 100}N_i\doteq 1.25\times 10^7$.
Using $N=1.25\times10^7$ with $\alpha=1.106$
and $B=1$ gives $n'\doteq 77.10$
raising the estimate only slightly from $72$.
The value $\alpha_*=1.106$ was chosen partly based on
fit and partly based on numerical convenience
that $\zeta(\alpha_*)=10$. Repeating our computation
with $1.05\le\alpha\le1.15$ gives values of 
$N$ that range from $1.1\times 10^7$ to $1.4\times 10^7$,
and estimates $n'$ from $71.04$ to $77.29$.
This estimate is very stable because the Zipf curve
is relatively straight in the critical region.

There is enough wiggling in the log--log plot Figure~\ref{fig:kilg}
between ranks $10$ and $50$, that can be attributed to $\e(X_i)$
not perfectly following a local power law there.  
The British National Corpus rank orderings
are not quite as reliable as those in the fitted Zipf--Poisson model.
Unsurprisingly, a one parameter model shows some lack of
fit on this enormous data set.

\begin{figure}[t]
\includegraphics[width=\hsize]{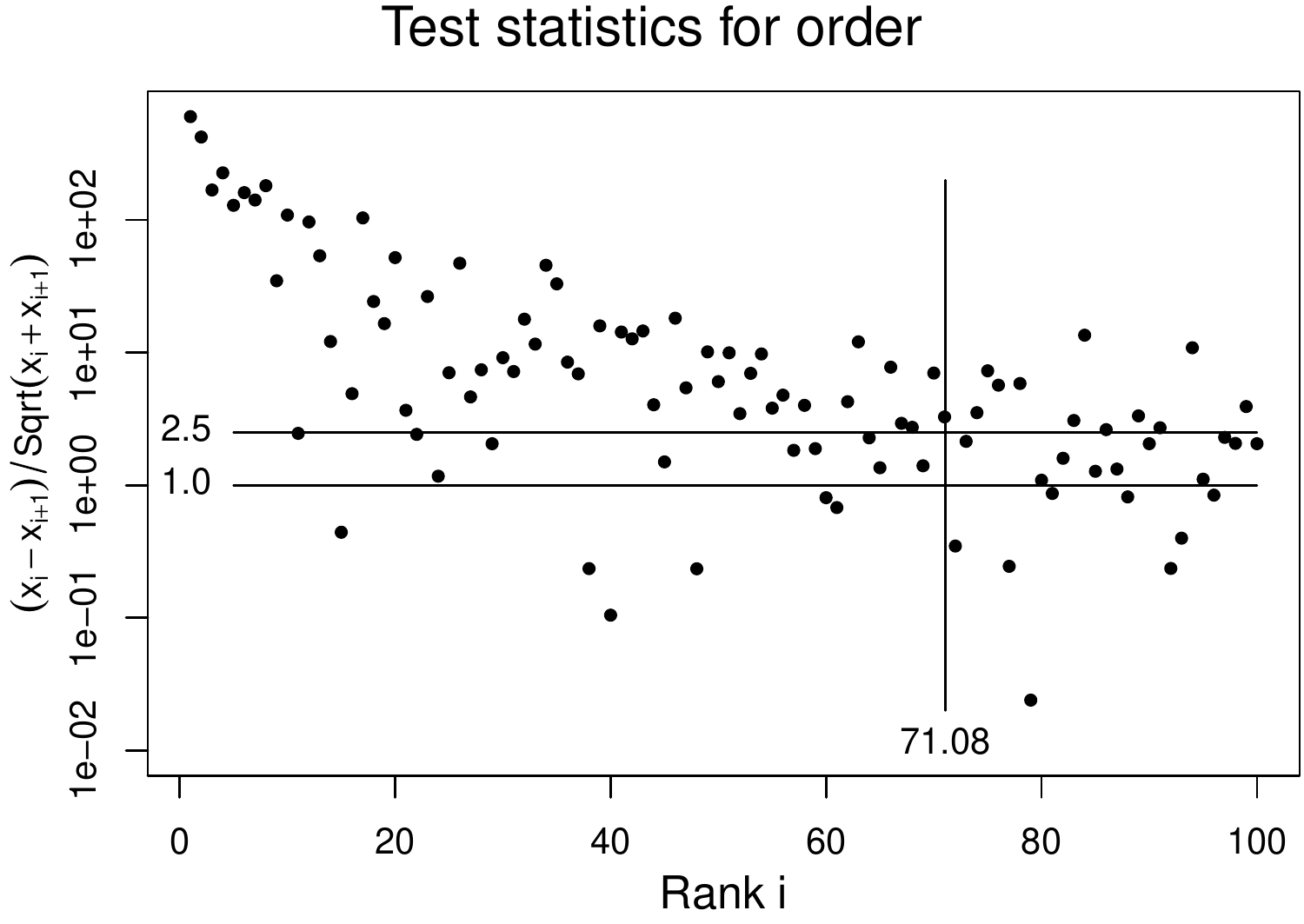}
\caption{\label{fig:sderrs}
This figure plots
$(X_{i}-X_{i+1})/\sqrt{X_{i}+X_{i+1}}$
versus $i=1,\dots,100$ for the BNC
data. The horizontal reference lines
are drawn at $2.5$ standard errors
and at $1.0$ standard errors.
The vertical line is drawn at $n'-1=71.08$.
}
\end{figure}

Figure~\ref{fig:sderrs} plots 
standard errors for the first $100$ consecutive
word comparisons. A horizontal line is at $2.5$.
The theorem predicts that the first $n'=72.08\doteq 72$
words would be correctly ordered
relative to each other.  When the first $n'$
words are correctly ordered the first $n'-1$
differences have the correct sign.
The vertical reference line is at $71.08$.
Beyond $72$ it is clear that many consecutive
word orderings are doubtful.
We also see a few small standard errors
for $n<72$ which correspond to some local flat spots
in Figure~\ref{fig:kilg}.
As a result we might expect a small number of transposition errors
among the first $72$
words, in addition to the large number
of rank errors that set in beyond the $72$'nd word,
as predicted by the Zipf--Poisson model.

\section{Proof of Theorem~\ref{thm:zpok}}\label{sec:proofs}

Theorem~\ref{thm:zpok} has three claims.
First, equation~\eqref{eq:zpok1} on correct ordering
of the $n$ most popular items within themselves,
follows from Corollary~\ref{cor:first} below.
Combining that corollary with Lemma~\ref{lem:youshallnotpass}
to rule out interlopers, establishes~\eqref{eq:zpok2}
in which the first $n$ items are correctly identified and ordered.
The third claim~\eqref{eq:zpok3}, showing the necessity
of the logarithmic factor, follows from Corollary~\ref{cor:second}.

\subsection{Some useful inequalities}
The proof of Theorem~\ref{thm:zpok} makes 
use of some bounds on Poisson
probabilities and the gamma function, collected here.

Let $Y\sim\poi(\lambda)$.  
\citet[page 485]{shor:well:1986} have
the following exponential bounds
\begin{align}\label{eq:klars}
\Pr( Y\ge t)&\le
\Bigl(1-\frac{\lambda}{t+1}\Bigr)^{-1}\frac{e^{-\lambda}\lambda^t}{t!}
\quad\text{for integers}\ t\ge \lambda,\quad\text{and}\\
\Pr( Y\le t)&\le
\Bigl(1-\frac{t}{\lambda}\Bigr)^{-1}\frac{e^{-\lambda}\lambda^t}{t!}
\quad\text{for integers}\ t<\lambda.\label{eq:lowertail}
\end{align}
\cite{klar:2000} shows that~\eqref{eq:klars}
holds for $t\ge \lambda-1$.
Equation \eqref{eq:klars} holds for real valued
$t\ge \lambda$  and
equation~\eqref{eq:lowertail} also holds
for real valued $t<\lambda$.
In both cases we interpret $t!$ as $\Gamma(t+1)$.

A classic result of \cite{teic:1955} is that
\begin{align}\label{eq:teichers}
\Pr( Y\le \lambda)\ge \exp(-1)
\end{align}
when $Y\sim\poi(\lambda)$. 
If $Y\sim\poi(\lambda)$, then
\begin{align}\label{eq:mich}
\sup_{-\infty <t<\infty}
\biggl|
\Pr\biggl(\frac{Y-\lambda}{\sqrt{\lambda}}\le t \biggr)
-\Phi(t)
\biggr|\le \frac{0.8}{\sqrt\lambda},
\end{align}
where $\Phi$ is the standard normal CDF.
Equation~\eqref{eq:mich} follows 
by specializing a Berry-Esseen
result for compound Poisson distributions 
\cite[Theorem 1]{mich:1993} to the case
of a Poisson distribution.

We will also use Gautschi's (1959) inequality \nocite{gaut:1959} on
the Gamma function,
\begin{align}\label{eq:gautschi}
x^{1-s}<\frac{\Gamma(x+1)}{\Gamma(x+s)} < (x+1)^{1-s}
\end{align} 
which holds for $x>0$ and $0<s<1$.



\subsection{Correct relative ordering, equation~\eqref{eq:zpok1}}

The difference of two independent Poisson
random variables has a \cite{skel:1946} distribution.
We begin with a Chernoff bound for the Skellam distribution.

\begin{lemma}\label{lem:skellbound}
Let $Z=X-Y$ where
$X\sim\poi(\lambda)$ and $Y\sim\poi(\nu)$
are independent and $\lambda\ge \nu$.
Then
\begin{equation}\label{eq:skellbound}
\Pr(Z\le 0)\le\exp\bigl( -(\sqrt{\lambda}-\sqrt{\nu})^2\bigr).
\end{equation}
\end{lemma}
\begin{proof}
Let $\varphi(t)=\lambda e^{-t}+\nu e^t$. Then $\varphi$
is a convex function attaining 
its minimum at $t^*=\log(\sqrt{\lambda/\nu})\ge 0$,
with $\varphi(t^*)=2\sqrt{\lambda\nu}$.
Using the Laplace transform of the Poisson distribution
$$m(t)\equiv \e(e^{-tZ})=e^{\lambda(e^{-t}-1)}e^{\nu(e^{t}-1)}
=e^{-(\lambda+\nu)}e^{\varphi(t)}.$$
For $t\ge 0$, Markov's inequality gives $\Pr(Z\le 0)=\Pr(e^{-tZ}\ge1)\le \e(e^{-tZ})$.
Taking $t=t^*$ 
yields~\eqref{eq:skellbound}.
\end{proof}

\begin{lemma}\label{lem:topn}
Let $X_i$ be sampled from the Zipf--Poisson ensemble.
Then for $n\ge 2$, 
\begin{align}\label{eq:topn}
\Pr(X_1>X_2>\dots>X_n)
\ge 1-n\exp\Bigl(-\frac{N\alpha^2}4\, n^{-\alpha-2}\Bigr).
\end{align}
\end{lemma}
\begin{proof}
By Lemma~\ref{lem:skellbound} and the Bonferroni inequality,
the probability that $X_{i+1}\ge X_i$ holds for any
$i<n$ is no more than
\begin{align}\label{eq:pwrong1}
\sum_{i=1}^{n-1}\exp\bigl( -\bigl(\sqrt{\lambda_i}-\sqrt{\lambda_{i+1}}\,\bigr)^2\,\bigr)
&= \sum_{i=1}^{n-1}\exp\bigl( -N\bigl(\sqrt{\theta_i}-\sqrt{\theta_{i+1}}\,\bigr)^2\,\bigr).
\end{align}
For $x\ge 1$, let $f(x)=x^{-\alpha/2}$.
Then $|\sqrt{\theta_i}-\sqrt{\theta_{i+1}}| = |f(i)-f(i+1)| = |f'(z)|$
for some $z\in(i,i+1)$.
Because $|f'|$ is decreasing, \eqref{eq:pwrong1} is at most
$n\exp(-Nf'(n)^2)$, establishing~\eqref{eq:topn}.
\end{proof}

Now we can establish the first
claim in Theorem~\ref{thm:zpok}.

\begin{corollary}\label{cor:first}
Let $X_i$ be sampled from the Zipf--Poisson ensemble.
Choose $n=n(N)\ge 2$ so that
$n\le (AN/\log(N))^{1/(\alpha+2)}$ holds
for all large enough $N$
where $A=\alpha^2(\alpha+2)/4$.
Then
$$\lim_{N\to\infty}\Pr( X_1>X_2>\dots>X_n)= 1.$$
\end{corollary}
\begin{proof}
For large enough $N$ we let
$n =(A_NN/\log(N))^{1/(\alpha+2)}$ for $A_N\le A$. Then
\begin{align*}
n\exp\Bigl(-\frac{N\alpha^2}4\, n^{-\alpha-2}\Bigr)
& =\biggl(\frac{A_NN}{\log(N)}\biggr)^{1/(\alpha+2)} \,N^{-{\alpha^2}/(4A_N)}\\
& =\biggl(\frac{A_NN}{\log(N)}\biggr)^{1/(\alpha+2)} \,N^{-{\alpha^2}/(4\alpha^2(\alpha+2)/4)}\\
&\le \biggl(\frac{A}{\log(N)}\biggr)^{1/(\alpha+2)}\\
&\to 0.
\end{align*}
The proof then follows from Lemma~\ref{lem:topn}.
\end{proof}

\subsection{Correct absolute ordering, equation~\eqref{eq:zpok2}}

For the second claim in Theorem~\ref{thm:zpok}
we need to control the probability that one
of the entities $X_i$ from the tail given by $i>n$,
can jump over one of the first $n$ entities.
Lemma~\ref{lem:jumper} bounds the probability that
an entity from the tail of the Zipf--Poisson ensemble
can jump over a high level $\tau$.

\begin{lemma}\label{lem:jumper}
Let $X_i$ for $i\ge 1$ be from the Zipf--Poisson
ensemble with parameter $\alpha>1$.
If $\tau\ge \lambda_n$
then
\begin{align}\label{eq:jumper}
\Pr\Bigl( \max_{i> n} X_i >\tau\Bigr)\le 
\frac{N^{1/\alpha}}\alpha\frac{\tau+1}{\tau+1-\lambda_n}\frac{{\tau}^{-1/\alpha}}{\tau-1/\alpha}.
\end{align}
\end{lemma}
\begin{proof}
First, 
$\Pr( \max_{i>n} X_i>\tau) \le
\sum_{i=n+1}^\infty\Pr(X_i>\tau)$ and then from \eqref{eq:klars}
\begin{align*}
\Pr\Bigl( \max_{i>n} X_i>\tau\Bigr)
\le & 
\Bigl(1-\frac{\lambda_n}{\tau+1}\Bigr)^{-1}
\sum_{i=n+1}^\infty\frac{e^{-\lambda_i}\lambda_i^\tau}{\Gamma(\tau+1)}.
\end{align*}
Now $\lambda_i = Ni^{-\alpha}$.
For $i>n$ we have $\tau>\lambda_i = Ni^{-\alpha}$.
Over this range, $e^{-\lambda}\lambda^\tau$ is an increasing
function of $\lambda$. 
Therefore,
\begin{align*}
\sum_{i=n+1}^\infty e^{-\lambda_i}\lambda_i^\tau
&\le\int_n^\infty e^{-Nx^{-\alpha}}(Nx^{-\alpha})^\tau \rd x\\
&\le\frac{N^{1/\alpha}}\alpha
\int_0^{Nn^{-\alpha}} e^{-y}y^{\tau-1/\alpha-1}\rd y\\
&\le\frac{N^{1/\alpha}}\alpha\Gamma(\tau-1/\alpha).
\end{align*}

As a result
$$
\Pr\Bigl(\max_{i>n}X_i>\tau\Bigr)\le
\frac{N^{1/\alpha}}\alpha
\frac{\tau+1}{\tau+1-\lambda_n}
\frac{\Gamma(\tau-1/\alpha)}{\Gamma(\tau+1)}.
$$
Now 
$$
\frac{\Gamma(\tau-1/\alpha)}{\Gamma(\tau+1)}
=\frac{\Gamma(\tau+1-1/\alpha)}{\Gamma(\tau+1)}
\frac1{\tau-1/\alpha}
< \frac{{\tau}^{-1/\alpha}}{\tau-1/\alpha}
$$ 
by Gautschi's inequality~\eqref{eq:gautschi}, 
with $s=1-1/\alpha$,
establishing~\eqref{eq:jumper}.
\end{proof}

For an incorrect ordering to arise, either an entity from
the tail exceeds a high level, or an entity from among the
first $n$ is unusually low.
Lemma~\ref{lem:youshallnotpass} uses a threshold
for which both such events are unlikely,
establishing the second claim~\eqref{eq:zpok2} of Theorem~\ref{thm:zpok}.

\begin{lemma}\label{lem:youshallnotpass}
Let $X_i$ for $i\ge 1$ be from the Zipf--Poisson ensemble with parameter
$\alpha>1$. 
Let $n(N)$ satisfy
$n\ge (AN/\log(N))^{1/(\alpha+2)}$ for $0<A<A(\alpha) = \alpha^2(\alpha+2)/4$.
Let $m\le (BN/\log(N))^{1/(\alpha+2)}$ for $0<B<A$.
Then
\begin{align}\label{eq:youshallnotpass}
\lim_{N\to\infty}\Pr\Bigl( \max_{i>n}X_i \ge X_m \Bigr)=0.
\end{align}
\end{lemma}
\begin{proof}
For any threshold $\tau$,
\begin{align}\label{eq:nopassterms}
\Pr\Bigl( &\max_{i>n}X_i\ge X_m\Bigr)
\le \Pr\Bigl(\max_{i>n}X_i>\tau\Bigr)+\Pr(X_m\le \tau).
\end{align}
The threshold we choose
is $\tau=\sqrt{\lambda_m\lambda_n}$
where $\lambda_i=\e(X_i)=Ni^{-\alpha}$.

Write $n = (A_NN/\log(N))^{1/(\alpha+2)}$
and $m= (B_NN/\log(N))^{1/(\alpha+2)}$
for $0<B_N< B <A_N<A<A(\alpha)$.
Then 
$\tau=\sqrt{\lambda_m\lambda_n}
= N(C_NN/\log(N))^{-\alpha/(\alpha+2)}$ 
where $C_N=\sqrt{A_NB_N}$.
Therefore
$$\tau= O\bigl( N^{{2}/(\alpha+2)}(\log(N))^{\alpha/(\alpha+2)}\bigr).$$

By construction,  $\tau> \lambda_n$ and so
by Lemma~\ref{lem:jumper} 
\begin{align*}
\Pr\Bigl(\max_{i>n}X_i>\tau\Bigr)&\le
\frac{N^{1/\alpha}}\alpha
\frac{\tau+1}{\tau+1-\lambda_n}\frac{{\tau}^{-1/\alpha}}{\tau-1/\alpha}.
\end{align*}
Because $\lambda_n/\tau = (B_N/A_N)^{\alpha/(2\alpha+4)}$,
we have $(\tau+1)/(\tau+1-\lambda_n)= O(1)$.
Therefore
\begin{align*}
\Pr\Bigl(\max_{i>n}X_i>\tau\Bigr)&
= O(N^{1/\alpha}\tau^{-1/\alpha-1})
= O\bigl(N^{-1/(\alpha+2)}(\log(N))^{(\alpha+1)/(\alpha+2)}\bigr)
\end{align*}
and so the first term in~\eqref{eq:nopassterms} tends to $0$
as $N\to\infty$.

For the second term in~\eqref{eq:nopassterms}, 
notice that $X_m$ has mean $\lambda_m> \tau$ and standard
deviation $\sqrt{\lambda_m}$.  
Letting $\rho = \alpha/(\alpha+2)$ and applying
Chebychev's inequality, we find that
\begin{align*}
\Pr(X_m\le \tau) & \le \frac{\lambda_m}{(\tau-\lambda_m)^2}\\
&= \frac{B_N^\rho}{(B_N^\rho-C_N^\rho)^2}N^{-2/(\alpha+2)}(\log(N))^{-\rho}\\
&\le \frac{1}{(A^{\rho/2}-B^{\rho/2})^2}N^{-2/(\alpha+2)}(\log(N))^{-\rho}\\
&\to 0
\end{align*}
as $N\to\infty$.
\end{proof}

Lemma~\ref{lem:youshallnotpass} is sharp enough for
our purposes.  A somewhat longer argument
in \cite{dyer:2010} shows that the interloper
phenomenon is ruled out even deeper into
the tail of the Zipf-Poisson ensemble.
Specifically, if $m\le(BN)^{\beta}$
and $n\ge (AN)^\beta$ for $0<B<A$ and $\beta<1/\alpha$,
then \eqref{eq:youshallnotpass} still holds.

\subsection{Limit to correct ordering, equation~\eqref{eq:zpok3}}

While we can get $(AN/\log(N))^{1/(\alpha+2)}$ entities properly
ordered, there is a limit to the number of correctly ordered
entities.  We cannot get above $CN^{1/(\alpha+2)}$ correctly
ordered entities, asymptotically. That is, the logarithm cannot be removed.
We begin with a lower bound on the probability
of a wrong ordering for two consecutive entities.

\begin{lemma}\label{lem:lowb}
Let $X_i$ be from the Zipf--Poisson ensemble with $\alpha>1$.
Suppose that $AN^{1/(\alpha+2)}\le i<i+1\le BN^{1/(\alpha+2)}$
where $0<A<B<\infty$.
Then for large enough $N$,
$$\Pr( X_{i+1}\ge X_i)\ge \frac13\,
\Phi\biggl(
-\alpha\,\frac{A^{\alpha/2}}{B^{\alpha+1}}
\biggr).$$
\end{lemma}
\begin{proof}
First
$
\Pr(X_{i+1}\ge X_i)\ge \Pr( X_{i+1} > \lambda_i)\Pr(X_i\le \lambda_i)
\ge \Pr(X_{i+1}>\lambda_i)/e
$
using Teicher's inequality~\eqref{eq:teichers}.
Next
$$\Pr(X_{i+1} > \lambda_i)=1-\Pr( X_{i+1}\le\lambda_i)
\ge \Phi\biggl(\frac{\lambda_{i+1}-\lambda_i}{\sqrt{\lambda_{i+1}}}\biggr)
-\frac{0.8}{\sqrt{\lambda_{i+1}}}.
$$
Now, 
\begin{align*}
\frac{\lambda_{i+1}-\lambda_i}{\sqrt{\lambda_{i+1}}}
=\sqrt{N}\,\frac{(i+1)^{-\alpha}-i^{-\alpha}}{\sqrt{(i+1)^{-\alpha}}}
=-\alpha\sqrt{N}\,\frac{(i+\eta)^{-\alpha-1}}{\sqrt{(i+1)^{-\alpha}}}
\end{align*}
for some $\eta\in(0,1)$.
Applying the bounds on $i$,
\begin{align*}
\frac{\lambda_{i+1}-\lambda_i}{\sqrt{\lambda_{i+1}}}
&\ge -\alpha\sqrt{N}
\frac{(N^{1/(\alpha+2)}A)^{\alpha/2}}{(N^{1/(\alpha+2)}B)^{\alpha+1}}
= -\alpha
\frac{A^{\alpha/2}}{B^{\alpha+1}}.
\end{align*}

Finally, letting $N\to\infty$ we have $\lambda_{i+1}\to\infty$
and so $0.8/\sqrt{\lambda_{i+1}}$ is eventually
smaller than $(1-e/3)\Phi(-\alpha A^{\alpha/2}B^{-\alpha-1})$.
Letting $\theta = -\alpha A^{\alpha/2}B^{-\alpha-1}$
we have, for large enough $N$,
\begin{align*}
\Pr( X_{i+1}\ge X_i) 
\ge \Bigl(\Phi(\theta)-\Bigl(1-\frac{e}3\Bigr)\Phi(\theta)\Bigr)\frac1e
= \frac13\Phi(\theta).
\qquad 
\qedhere 
\end{align*}
\end{proof}

To complete the proof of Theorem~\ref{thm:zpok}
we establish equation~\eqref{eq:zpok3}.
For $n$ beyond a multiple of $N^{1/(\alpha+2)}$,
the reverse orderings predicted by Lemma~\ref{lem:lowb}
cannot be avoided.

\begin{corollary}\label{cor:second}
Let $X_i$ be sampled from the Zipf--Poisson ensemble.
Suppose that $n=n(N)$ satisfies
$n\ge CN^{1/(\alpha+2)}$ for $0<C<\infty$. 
Then
$$\lim_{N\to\infty}\Pr( X_1>X_2>\dots>X_n)= 0.$$
\end{corollary}
\begin{proof}
Let $p\in(0,1)$ be a constant such that
$\Pr(X_{i+1}\ge X_i)\ge p$ holds for 
all large enough $N$
and $(C/2)N^{1/(2+\alpha)}\le i<i+1\le CN^{1/(2+\alpha)}$.
For instance Lemma~\ref{lem:lowb} shows
that $p=\Phi( -\alpha (C/2)^{\alpha/2}/C^{\alpha})/3
=\Phi( -\alpha (2C)^{-\alpha/2})/3$ is such a constant.
Then
\begin{align}\label{eq:pairs}
\Pr( X_1>X_2>\dots>X_n) \le 
\sideset{}{^*}\prod_i\Pr(X_i>X_{i+1})
\end{align}
holds where $\prod^*$ is over all odd integers
$i\in[(C/2)N^{1/(\alpha+2)},CN^{1/(\alpha+2)})$.
There are roughly $CN^{1/(\alpha+2)}/4$ odd integers
in the product.
For large enough $N$, the right side of~\eqref{eq:pairs}
is below $(1-p)^{CN^{1/(\alpha+2)}/5}\to 0$.
\end{proof}

\section{Discussion}\label{sec:discussion}

We have found that the top few entities
in a Zipf plot of counts can be expected
to be in the correct order, when their
frequencies are measured with Poisson errors.
Even in the idealized Zipf setting, the
number of correctly ordered entities
grows fairly slowly with $N$.

Our transition point is at
$n'= (\alpha^2(\alpha+2)N/(4\log(N))^{1/(\alpha+2)}$ 
and estimating $N$ from $T=\sum_iX_i$
leads to the estimate
$$
\hat n = \biggl(\frac{\alpha^2(\alpha+2)T/\zeta(\alpha)}
{4\log(T/\zeta(\alpha))}
\biggr)^{\frac1{\alpha+2}}.
$$

The threshold $n'$ uses some slightly
conservative estimates to get a rate in $N$.
For the Zipf--Poisson ensemble with $N=10^7$
and $\alpha=1.106$
we can use \eqref{eq:topn} of Lemma~\ref{lem:topn}
directly to  find
\begin{align*}
1-\Pr(X_1>X_2>\cdots>X_{72}) \le 
\sum_{i=1}^{71}\exp(-N(i^{-\alpha/2}-(i+1)^{-\alpha/2})^2)\doteq 0.0199.
\end{align*}
We get a bound of $1$\% by taking $n=70$
and a bound of $5$\% by taking $n=76$.
The formula for $\hat n$ comes remarkably
close to what we get working directly
with equation~\eqref{eq:topn}.

The Skellam bounds do not assume a Zipf rate
for the Poisson means. Therefore we can use
them to generalize the computation above.
For example, with a Zipf--Mandelbrot--Poisson ensemble
having $X_i\sim \poi(N(i+k)^{-\alpha})$
we can still apply 
equation~\eqref{eq:topn} 
to show that the probability of an error
among the first $n$ ranks is at most
\begin{align}\label{eq:pofnNalphak}
p(n;N,\alpha,k)=\sum_{i=1}^{n-1}\exp\Bigl( -N\Bigl( (i+k)^{-\alpha/2}
-(i+k+1)^{-\alpha/2}\Bigr)^2\Bigr).
\end{align}
A conservative estimate of the number of correct
positions in the Zipf--Mandelbrot--Poisson ensemble is
\begin{align}\label{eq:pickn}
n' = \max\{ n\ge 1\mid p(n;N,\alpha,k)\le 0.01\}
\end{align}
with $n'=0$ if $p(1;N,\alpha,k)>0.01$.
We can estimate $N$ by 
$T/\zeta(\alpha,k-1)$
where $T=\sum_iX_i$ and
$\zeta(\alpha,h)=\sum_{\ell=0}^\infty (\ell + h)^\alpha$
is the Hurwitz zeta function.

Equation~\eqref{eq:pickn} is conservative
because it stems from the Bonferroni inequality,
and does not adjust for two or more order
relations being violated.
It will be less conservative for small
target probabilities like $0.01$ than for large ones
where adjustments are relatively more important.

Our focus is on the ranks that are correctly
estimated. Methods to estimate parameters
of the Zipf distribution or Zipf--Mandelbrot
distribution typically use values of $X_i$
for $i$ much larger than the number
of correctly identified items. It is not unreasonable
to do so, because ordering errors
tend to distribute the values of $X_i$
both above and below the parametric trend line.

A small number of correct unique words
can correspond to a reasonably large
fraction of word usage.  The BNC is 
roughly $6.2$\% `the' and the top $72$
words comprise about $45.3$\% of the
corpus.

For large $N$, the top $n_\epsilon=N^{1/(\alpha+2)-\epsilon}$
entities get properly ordered with very high probability
for $0<\epsilon<1/(\alpha+2)$.
The tail beyond $n_\epsilon$ accounts
for a proportion of
data close to $\zeta(\alpha)^{-1}\int_{n_\epsilon}^\infty
x^{-\alpha}\rd x
= O(n_\epsilon^{-\alpha+1})
= O(N^{(1-\alpha)/(\alpha+2)+\epsilon'})
$
for $\epsilon' = \epsilon(\alpha-1)$.
Taking small $\epsilon$ and recalling
that $\alpha>1$ we find that the fraction
of data from improperly ordered entities
vanishes in the Zipf--Poisson ensemble.
When $\alpha$ is just
barely larger than $1$ the rate may be slow.

\section*{Acknowledgments}

This work was supported by the U.S.\ National
Science Foundation under grant DMS-0906056 and by a
National Science Foundation Graduate Research Fellowship.

\bibliographystyle{plainnat}
\bibliography{transpo}

\end{document}